\documentclass{article}
\usepackage{latexsym}
\usepackage{amsmath}
\usepackage{amssymb}
\usepackage{makeidx}
\usepackage{url}
\usepackage{graphicx}
\usepackage{colortbl}
\usepackage{pst-all}
\usepackage{subfigure}
\usepackage{array}
\usepackage[left=4cm,top=3cm,right=4cm,bottom=3cm,nohead]{geometry}
\usepackage{hyperref}
\usepackage{authblk}

\newtheorem{definition}{Definition}[section]
\newtheorem{proposition}[definition]{Proposition}

\newtheorem{example}[definition]{Example}
\newenvironment{proof}{\noindent\textbf{Proof}\itshape}{~\hfill~$\Box$ \newline\newline}

\def\after{\mbox{\textnormal{\sf After}}}

\def\datalog{\mbox{\textsc{datalog}}}

\def\dom{\mbox{\textnormal{dom}}}
\def\edb{\mbox{\textit{EDB}}}

\def\ic{\mbox{\textit{IC}}}
\def\idb{\mbox{\textit{IDB}}}

\def\metaequiv{\Leftrightarrow}

\def\optimize{\mbox{\textnormal{\sf Optimize}}}

\def\param#1{{\mathbf #1}}

\def\provfbr{\mbox{$\;\vdash_R\;$}}

\def\rewritenospace{\Rightarrow}
\def\rewrite{\;\; \rewritenospace \;\;}
\def\simplanguage{{\cal L}_{\mbox{\tiny{\sf S}}}}
\def\simpextlanguage{{\cal L}_{\mbox{\tiny{{\sf H}}}}}

\def\simp{\mbox{\textnormal{\sf Simp}}}

\def\true{\mbox{\textit{true}}}

\def\unfold{\mbox{\textnormal{\sf Unfold}}}
\def\vars{\mbox{\textnormal{vars}}}
\def\vec{\mathaccent"017E }

\providecommand{\keywords}[1]{\textbf{\textit{Keywords---}} #1}

\begin{document}

\title{Simplified integrity checking for an expressive class of denial constraints}

\author{Davide Martinenghi}

\affil{Politecnico di Milano, DEIB\\Piazza Leonardo 32, 20133 Milan, Italy.\\email:
\url{davide.martinenghi@polimi.it}
}

\date{}

\maketitle

\begin{abstract}
Data integrity is crucial for ensuring data correctness and quality, maintained through integrity constraints that must be continuously checked, especially in data-intensive systems like OLTP. While DBMSs handle common constraints well, complex constraints often require ad-hoc solutions. Research since the 1980s has focused on automatic and simplified integrity constraint checking, leveraging the assumption that databases are consistent before updates. This paper discusses using program transformation operators to generate simplified integrity constraints, focusing on complex constraints expressed in denial form. In particular, we target a class of integrity constraints, called extended denials, which are more general than tuple-generating dependencies and equality-generating dependencies.
These techniques can be readily applied to standard database practices and can be directly translated into SQL.
\end{abstract}


\pagestyle{plain}

\section{Introduction}

Data integrity is a ubiquitous concern that is at the basis of crucial objectives such as data correctness and data quality.
Typically, a specification of the correct states of a database is given through integrity constraints, which can be thought of as formulas that need to be kept satisfied throughout the life of the database.
In any data-intensive scenario, such as an OLTP system, date are updated all the time and, therefore, integrity constraints need to be checked continuously. While DBMSs are able to optimally handle common constraints such as key and foreign key constraints, the standard practice is to resort to ad-hoc solutions (for instance, at the application level or through manually designed triggers) when the integrity constraints to be maintained are more complex.
Automatic integrity constraint checking, and, in particular, simplified integrity constraint checking is a topic that has spurred numerous research attempts since the 1980s with the objective of mechanically producing the best possible set of conditions to check against a database in the face of updates.
The main idea, common to the majority of the approaches, is to exploit the assumption that a database is \emph{consistent} (i.e., satisfies the integrity constraints) before an update; with this, the formula to be checked, corresponding to the original integrity constraints, may be simplified so as to avoid checks of subformulas that are already guaranteed to hold thanks to the consistency assumption.

For instance, consider the relation $b(I,T)$, indicating that there is a book with an ISBN $I$ and a title $T$. The fact that the ISBN univocally identifies a book can be expressed as the integrity constraint $T_1=T_2\leftarrow b(I,T_1) \land b(I,T_2)$ (with all variables implicitly universally quantified). Now, assume that the tuple $b(i,t)$ is added to the database. Under the consistency assumption, one just needs to check that there is no other book with the same ISBN $i$ and a different title than $t$,\footnote{We are using set semantics, so having more than one occurrence of $b(i,t)$ is harmless.} which could be expressed by the simplified formula $T=t\leftarrow b(i,T)$. Without the consistency assumption, instead, we would have had to check satisfaction of the original constraint for all the books already present in the database.

In this paper, we discuss the use of program transformation operators to the generation of simplified integrity constraints. We focus on databases and constraints for which we can implement these operators in terms of rewrite rules based on resolution, subsumption and replacement of specific patterns.
In particular, we concentrate on integrity constraints that are more expressive than tuple-generating dependencies (TGDs) and equality-generating dependencies (EGDs), which we express in denial form with negated existential quantifiers, possibly nested.

The techniques described here can be immediately applied to standard database practice, since our logical notation can be easily translated into SQL.
Our results are mainly taken from an unpublished chapter of the Ph.D. dissertation~\cite{M:PHD2005}. To the best of our knowledge, to date there is no other approach that has considered such an expressive class of constraints for the problem of simplified integrity constraint checking.

\section{Relational and deductive databases}\label{cha:relational-and-deductive}

In this chapter we present the basic notation related to deductive databases, and use the syntax of the $\datalog$ language as a basis to formulate the main concepts.
We refer to standard texts in logic, logic programming and databases
such as \cite{DBLP:books/cs/Ullman88,NM95} for further
discussion.
As a notational convention, we use lowercase letters to denote \emph{predicates} ($p$, $q$,
$\ldots$) and \emph{constants} ($a$, $b$, $\ldots$) and uppercase letters
($X$, $Y$, $\ldots$) to denote \emph{variables}.
A \emph{term} is either a variable or a constant; conventionally, terms are also denoted by lowercase letters ($t$, $s$, $\ldots$) and sequences of terms are indicated by
vector notation, e.g., $\vec t$.
The language also includes standard connectives ($\land$, $\lor$, $\lnot$, $\leftarrow$), quantifiers, and punctuation.

In quantified formulas such as $(\forall X F)$ and $(\exists X F)$, $X$ is said to be
\emph{bound}; a variable not bound is said to occur \emph{free}. A formula with
no free variables is said to be
\textit{closed}.
Free variables are indicated in
boldface ($\param a$, $\param b$, $\dots$) and referred to as
\emph{parameters}\label{page:parameter-introduced};
the other variables, unless otherwise stated, are
assumed to be universally quantified.

Formulas of the form $p(t_1,\ldots, t_n)$ are called
\emph{atoms}. An atom preceded by a $\lnot$ symbol is a
\emph{negated atom}; a
\emph{literal} is either an atom or a negated atom.

Predicates are divided into three pairwise disjoint sets:
\emph{intensional}\index{predicate!intensional predicate},
\emph{extensional}\index{predicate!extensional predicate}, and
\emph{built-in}\index{predicate!built-in predicate} predicates.
Intensional and extensional predicates are collectively called
\emph{database predicates}\index{predicate!database predicate};
atoms and literals are classified similarly according to their
predicate symbol. There is one built-in binary predicate for term
equality ($=$)\index{$=$ (equality)}\index{term!equality}, written
using infix notation;
$t_1\neq t_2$\index{$\neq$ (non-equality)}\index{term!non-equality}
is a shorthand for $\lnot (t_1 = t_2)$ for any two terms $t_1$,
$t_2$.

A \emph{substitution} $\{\vec X/\vec t\}$ is a mapping from
the variables in $\vec X$ to the terms in $\vec t$; its application is the expression arising when each
occurrence of a variable in $\vec X$ is simultaneously replaced by
the corresponding term in $\vec t$.
A formula or term which contains no
variables is called \emph{ground}. A
substitution $\{\vec X/\vec Y\}$ is called a
\emph{renaming} iff $\vec Y$ is a permutation of
$\vec X$. Formulas $F$, $G$ are \emph{variants} of
one another if $F=G\rho$ for some renaming $\rho$.
A substitution $\sigma$ is said to be more general than $\theta$ iff there exists a substitution $\eta$ such that $\theta=\sigma\eta$.
A \emph{unifier} of 
$t_1,\dots,t_n$ is a substitution $\sigma$ such that
$t_1\sigma=\dots=t_n\sigma$;
$\sigma$ is a \emph{most general unifier} (mgu) of
$t_1,\dots,t_n$ if it is more general than any other unifier thereof.

A \emph{clause} is a disjunction of literals $L_1\lor\dots\lor L_n$.
In the context of deductive databases, clauses are expressed in implicational form,
$A\leftarrow L_1\land\cdots\land L_n$,
where $A$ is an atom and $L_1,\ldots,L_n$ are literals; $A$ is
called the \emph{head} and $L_1\land\cdots\land
L_n$ the \emph{body} of the clause. The head is optional and when it is omitted the
clause is called a \emph{denial}.
A \emph{rule} is a clause whose head is intensional, and
a \emph{fact} is a clause whose head is extensional and
ground and whose body is empty (understood as $\true$). Two clauses
are \emph{standardized apart} if they have no bound variable in common.
As customary, we also assume
clauses to be \emph{range
restricted}, i.e., each variable must occur in a positive database literal in the body.

Deductive databases are characterized by three components: facts,
rules and integrity constraints. An \emph{integrity
constraint} can, in general, be any
(closed) formula. In the context of deductive databases, however, it
is customary to express integrity constraints in some canonical
form; we adopt here the denial form, which gives a clear indication
of what must not occur in the database.

\begin{definition}[Schema, database]\label{def:schema-database}
A database {\em schema} $S$ is a pair
$\langle \idb, \ic\rangle$, where $\idb$ (the \emph{intensional
database}) is a finite set of
range restricted rules, and $\ic$ (the {\em constraint
theory}\index{constraint theory}) is a finite set of denials.
A {\em database} $D$ on $S$ is a pair $\langle \idb,
\edb\rangle$, where $\edb$ (the \emph{extensional
database}) is a finite set of
facts;
$D$ is said to be {\em based on} $\idb$.
\end{definition}
If the $\idb$ is understood, the database is identified with
$\edb$ and the schema with $\ic$.

For the semantics of a deductive database, we refer to standard texts, such as~\cite{Llo87}.
There are database classes for which the existence
of a unique \emph{intended} minimal
model is guaranteed. One such class is that of \emph{stratified databases}, i.e., those in which no predicate depends negatively on itself. This can be checked syntactically by building a graph with a node for each predicate and an arc between $p$ and $q$ each time $p$ is in the body and $q$ in the head of the same clause; the arc is labeled with a ``$-$'' if $p$ occurs negated. The database is stratified if the graph has no cycle including arcs labeled with a ``$-$''.

Stratified databases admit a unique ``intended''
minimal model,
called the \emph{standard model}~\cite{DBLP:books/mk/minker88/AptBW88,DBLP:journals/jlp/AptB94}.
Other minimal models are known, e.g., the \emph{perfect
model} for \emph{locally stratified
deductive databases} \cite{Prz88}, the \emph{stable
model} semantics \cite{GL88} and the
\emph{well-founded model} semantics
\cite{GRS88}.
We adopt the semantics of stratified databases and write $D\models \phi$, where $D$ is a database and $\phi$ a
closed formula, to indicate that $\phi$ holds in $D$'s standard model.

We now introduce a notation for querying and updating a database.
A \emph{query}\index{query} is an expression of the form $\Leftarrow
A$, where $A$ is an atom whose predicate is intensional.
When no ambiguity arises, a given query may be indicated
directly by means of its \emph{defining
formula} (instead of the predicate
name).
\begin{definition}[Update]\label{def:update}\index{update}
An {\em update}\index{update!predicate update} for an
extensional predicate $p$ in a database $D$ is an expression of the
form $p(\vec X) \Leftarrow p'(\vec X)$ where $\Leftarrow p'(\vec X)$
is a query for $D$. For an update $U$, the {\em
updated}\index{database!updated database} database $D^U$ has the same EDB as $D$ 
but in which, for every extensional predicate $p$ updated as $p(\vec X) \Leftarrow p'(\vec X)$ in $U$, the subset
$\{p(\vec t) \mid D \models p(\vec t)\}$ of $\edb$ is replaced by
the set $\{p(\vec t) \mid D \models p'(\vec t)\}$.
\end{definition}

\begin{example}\label{ex:updates}
Consider $\idb_1 = \{p'(X)\leftarrow
p(X),\;\;\;\; p'(X)\leftarrow X= a\}.$
The following update $U_1$ describes the addition of fact $p(a)$:
$\{p(X)\Leftarrow p'(X)\}.$
For convenience, wherever possible, defining formulas instead of queries
will be written in the body of predicate updates.
Update $U_1$ will, e.g., be indicated as follows:
$\{p(X) \Leftarrow p(X) \lor X= a\}.$
\end{example}

The constraint verification problem may be formulated as follows.
Given a database $D$, a constraint theory $\Gamma$ such that
$D \models \Gamma$, and an update $U$, does $D^U\models\Gamma$ hold
too?

Since checking directly whether $D^U \models \Gamma$ holds may be
too expensive, we aim to obtain a constraint theory $\Gamma^U$, called a \emph{pre-test}, such
that $D^U\models \Gamma$ iff $D\models \Gamma^U$ and $\Gamma^U$ is
easier to evaluate than $\Gamma$.

\section{A simplification procedure for obtaining a pre-test}\label{cha:procedure}
We report here the main steps of the simplification procedure described in~\cite{M:PHD2005}, whose operators will be used as building blocks in the next section to handle an expressive class of integrity constraints.
The procedure consists of two main steps: $\after$ and $\optimize$.
The former builds a schema that, while applying to the current state, represents the given integrity constraints in the state after the update. The latter eliminates
redundancies from $\after$'s output and exploits the consistency assumption.

A mere condition for checking integrity without actually executing the update is given by the notion of weakest precondition~\cite{D76,Qia88,H69}.
\begin{definition}[Weakest precondition]
Let $S=\langle\idb,\Gamma\rangle$ be a schema and $U$ an update. A
schema $S'=\langle\idb',\Gamma'\rangle$ is a
\emph{weakest precondition (WP)}\index{weakest
precondition}\index{2wp@WP (weakest precondition)} of $S$ with
respect to $U$ whenever $D\models\Gamma'\metaequiv D^U\models\Gamma$
for any database $D$ based on $\idb\cup\idb'$.
\end{definition}
To employ the consistency assumption, we extend the class of
checks of interest as follows.

\begin{definition}[Conditional weakest
precondition]\label{def:conditiona-weakest-precondition}
Let $S=\langle\idb,\Gamma\rangle$ be a schema and $U$ an update. Schema $S'=\langle\idb',\Gamma'\rangle$
 is a \emph{conditional
weakest precondition (CWP)}\index{conditional weakest
precondition}\index{2cwp@CWP (conditional weakest precondition)} of
$S$ wrt $U$ whenever $D\models\Gamma'\metaequiv
D^U\models\Gamma$ for any database $D$ such that $D\models\Gamma$.
\end{definition}
A WP is also a CWP
but the reverse does not necessarily hold.
\begin{example}
Consider the update and IDB from Example \ref{ex:updates} and let
$\Gamma_1$ be the constraint theory $\{\leftarrow p(X)\land q(X)\}$
stating that $p$ and $q$ are mutually exclusive. Then $\{\leftarrow
q(a)\}$ is a CWP (but not a WP) of $\Gamma_1$ with respect to $U_1$.
\end{example}

We now formally characterize a language, that we
call $\simplanguage$, on which a core simplification procedure can
be applied.
We introduce now a few technical notions that are needed in order to
precisely identify the class of predicates, integrity constraints
and updates that are part of $\simplanguage$, which is a
non-recursive, function-free language equipped with negation.
\begin{definition}[Starred dependency graph]\label{def:dependency-graph-starred}
Let $S=\langle \idb, \Gamma\rangle$ be a sche\-ma in which the
$\idb$ consists of a set of disjunctive (range restricted) predicate
definitions and $\Gamma$ is a set of range restricted denials. Let
$\cal G$ be a graph that has a node $N^p$ for each predicate $p$ in
$S$ plus another node named $\perp$, and no other node; if $p$'s
defining formula has variables not occurring in the head, then $N^p$ is
marked with a ``*''. For any two predicates $p$ and $p'$ in $S$,
$\cal G$ has an arc from $N^{p}$ to $N^{p'}$ for each occurrence of
$p$ in an $\idb$ rule in which $p'$ occurs in the head; similarly,
there is an arc from $N^{p}$ to $\perp$ for each occurrence of $p$
in the body of a denial in $\Gamma$. In both cases the arc is
labelled with a ``$-$'' (and said to be \emph{negative}) iff $p$
occurs negatively. $\cal G$ is the \emph{starred dependency
graph}\index{starred dependency graph}\index{dependency
graph!starred dependency graph} for $S$.
\end{definition}
\begin{example}\label{ex:dependency-graph-starred}
Let $S_1=\langle\idb_1,\Gamma_1\rangle$ and
$S_2=\langle\idb_2,\Gamma_2\rangle$ be schemata with
$$\begin{array}{rll}
\idb_1 & = \{ &s_1(X) \leftarrow r_1(X,Y) \},\\
\Gamma_1 &= \{&\leftarrow p_1(X) \land \lnot s_1(X)\},\\
\idb_2 & = \{ &s_2(X) \leftarrow r_2(X,Y),\\
&&q_2(X) \leftarrow \lnot s_2(X)\land t_2(X)\},\\
\Gamma_2 &= \{&\leftarrow p_2(X) \land \lnot q_2(X)\}.\\
\end{array}$$
The starred dependency graphs of $S_1$ and $S_2$ are shown in Figure
\ref{fig:dependency-graph}.
\end{example}
\begin{figure}[t]
\centerline{\setlength{\unitlength}{1cm}
\begin{picture}(0,0)
\put(-5,0){${\cal G}_1$}
\put(-4,0){\circle{.5}}
\put(-4.15,-.10){$\perp$}
\put(-3.6,.1){$-$}
\put(-3,0){\circle{.5}}
\put(-3.15,-.10){$s_1^*$}
\put(-4,-1){\circle{.5}}
\put(-4.1,-1.10){$p_1$}
\put(-2,0){\circle{.5}}
\put(-2.1,-.1){$r_1$}
\put(-3.25,0){\vector(-1,0){.5}}
\put(-2.25,0){\vector(-1,0){.5}}
\put(-4,-.75){\vector(0,1){.5}}
\put(3.5,0){${\cal G}_2$}
\put(0,0){\circle{.5}}
\put(-0.15,-.10){$\perp$}
\put(.4,.1){$-$}
\put(1.4,.1){$-$}
\put(1,0){\circle{.5}}
\put(.90,-.10){$q_2$}
\put(2,0){\circle{.5}}
\put(1.85,-.10){$s_2^*$}
\put(0,-1){\circle{.5}}
\put(-.10,-1.10){$p_2$}
\put(3,0){\circle{.5}}
\put(2.9,-.1){$r_2$}
\put(.75,0){\vector(-1,0){.5}}
\put(1.75,0){\vector(-1,0){.5}}
\put(2.75,0){\vector(-1,0){.5}}
\put(0,-.75){\vector(0,1){.5}}
\put(0.85,-1.10){$t_2$}
\put(1,-1){\circle{.5}}
\put(1,-0.75){\vector(0,1){.5}}
\end{picture}}\vspace{1cm}
  \caption{The starred dependency graphs for the schemata of Example~\ref{ex:dependency-graph-starred}.\label{fig:dependency-graph}}
\end{figure}
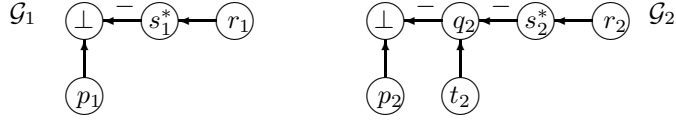
\begin{definition}[$\simplanguage$]\label{def:simp-language}\index{$\simplanguage$}
A schema 
$S$ is in
$\simplanguage$ if
its starred dependency graph $\cal G$ is acyclic and
in every path in $\cal G$ from a starred node to $\perp$ the
number of arcs labelled with ``$-$'' is even.
An update $\{p_1(\vec
X_1)\Leftarrow q_1(\vec X_1), \; \dots, \; p_m(\vec X_m)\Leftarrow
q_m(\vec X_m) \}$
is in $\simplanguage$ if the graph
obtained from $\cal G$ by adding an arc from $N^{q_i}$ to $N^{p_i}$,
$1\leq i \leq m$, meets the above condition.
\end{definition}
Acyclicity corresponds to absence of recursion.
The second condition requires
that the unfolding of the intensional predicates in the constraint
theory does not introduce any
negated existentially quantified variable. In
particular, starred nodes correspond to literals with
existentially quantified variable (called \emph{non-distinguished}).
In order to avoid negated existential quantifiers, the number of ``$-$'' signs must be even, as an even number of
negations means no negation. This requirement will be removed in the next section.
\begin{example}\label{ex:updated-dependency-graph}
Consider $S_1$ and $S_2$ from Example
\ref{ex:dependency-graph-starred}. Clearly $S_1$ is not in
$\simplanguage$, as in its starred dependency graph ${\cal G}_1$ in
Figure \ref{fig:dependency-graph} there is a path from $s_1^*$ to
$\perp$ containing one ``$-$'' arc, while $S_2$ is in
$\simplanguage$, as in ${\cal G}_2$ the only path from $s_2^*$ to
$\perp$ contains two ``$-$'' arcs.
\end{example}

\subsection{Generating weakest
preconditions}\label{sec:generating-wp}
The following syntactic transformation, $\after$, generates a WP.
\begin{definition}[$\after$]\label{def:after-general}\index{1after@$\after$}
Let $S=\langle\idb,\Gamma\rangle$ be a schema and $U$ an update $p(\vec X) \Leftarrow p^U(\vec X)$.
\begin{itemize}
\item Let us indicate with $\Gamma^U$ a copy of $\Gamma$ in which
any atom $p(\vec t)$ is simultaneously
replaced by the expression $p^U(\vec t)$ and every intensional
predicate $q$ is simultaneously replaced by a new intensional
predicate $q^U$ defined in $\idb^U$ below.
\item Similarly, let us indicate with $\idb^U$ a copy of $\idb$ in
which the same replacements are simultaneously made,
and let $\idb^*$ be the biggest subset of $\idb\cup\idb^U$ including
only definitions of predicates on which $\Gamma^U$ depends.
\end{itemize}
We define $\after^U(S) = \langle\idb^*,\Gamma^U\rangle$.
\end{definition}
\begin{example}\label{ex:after}
Consider the updates and $\idb$ definitions of
Example~\ref{ex:updates}. Let schema
$S_1$ be $\langle\idb_1,\Gamma_1\rangle$, where $\Gamma_1
=\{\leftarrow p(X)\land q(X)\}$ states that $p$ and $q$ are mutually
exclusive. We have
$\after^{U_1}(S_1)=\langle\idb_1,\Gamma_1^{U_1}\rangle$, where
$\Gamma_1^{U_1} =\; \{\leftarrow p'(X)\land q(X)\}.$
We replace $p'(X)$ by its defining formula
and thus omit the intensional database:
$$\after^{U_1}(\Gamma_1)=\Sigma = \begin{array}[t]{clll}
      \{ & \leftarrow & p(X)\land q(X), & \\
         & \leftarrow & X=a \land q(X) & \}.
\end{array}$$
\end{example}

Clearly, $\after^U(S)$ is a WP of $S$ with respect to $U$, and, trivially, a CWP.

For $\simplanguage$, we use \emph{unfolding}~\cite{DBLP:journals/jlp/Bol93a} to replace every
intensional predicate by its definition until only extensional
predicates appear. To this end, we use the
$\unfold_{\simplanguage}$ operator below.
\begin{definition}[Unfolding]\label{def:unfolding}\index{1unfold@$\unfold_{\simplanguage}$}
Let $S=\langle \idb, \Gamma\rangle$ be a database schema in
$\simplanguage$. Then, $\unfold_{\simplanguage}(S)$ is the schema
$\langle\emptyset,\Gamma'\rangle$, where $\Gamma'$ is the set of
denials obtained by iterating the two following steps as long as
possible:
\begin{enumerate}
\item replace, in $\Gamma$, each atom $p(\vec t)$ by $F^p \{\vec X / \vec t\}$, where 
$F^p$ is $p$'s defining formula and $\vec X$ its head
variables. If no replacement was made, then stop;
\item transform the result into a set of denials
according to the following patterns:
\begin{itemize}
\item  $\leftarrow A\land(B_1\lor B_2)$ is replaced by $\leftarrow
A\land B_1$ and $\leftarrow A\land B_2$;
\item  $\leftarrow A\land\lnot(B_1\lor B_2)$ is replaced by
$\leftarrow A\land\lnot B_1\land \lnot B_2$;
\item  $\leftarrow A\land\lnot(B_1\land B_2)$ is replaced by
$\leftarrow A\land \lnot B_1$ and $\leftarrow A\land \lnot B_2$.
\end{itemize}
\end{enumerate}
\end{definition}
Due to the implicit outermost universal quantification of the
variables, non-distinguished variables in a predicate definition are
existentially quantified to the right-hand side of the arrow, as
shown in example \ref{ex:non-distinguished-vars} below. For this
reason, with no indication of the quantifiers, the replacements in
definition \ref{def:unfolding} preserve equivalence iff no predicate
containing non-distinguished variables occurs negated in the
resulting expression.
\begin{example}\label{ex:non-distinguished-vars}
Consider $S_1$ from Example~\ref{ex:dependency-graph-starred},
which, as shown, is not in $\simplanguage$. With the explicit
indication of the quantifiers we have
$ \idb_1 \equiv \{ \forall X (s_1(X) \leftarrow \exists Y \;
r_1(X,Y)) \}.$
The replacement, in $\Gamma_1$, of $s_1(X)$ by its definition in
$\idb_1$ would determine the formula
$\Gamma_1' = \{\forall X,Y(\leftarrow p_1(X) \land \lnot
r_1(X,Y))\}.$
However, this replacement is not equivalence-preserving, because a
predicate ($r_1$) containing a non-distinguished variable occurs
negated:
$\Gamma_1 \equiv \{\forall X(\leftarrow p_1(X) \land \lnot \exists
Y\; r_1(X,Y))\} \not\equiv \Gamma_1'.$
\end{example}
We shall cover the cases in which $\lnot \exists$ may
occur in denials in Section \ref{cha:extensions}.

The language $\simplanguage$ is closed under unfolding and
$\unfold_{\simplanguage}$ preserves equivalence.
We then refine $\after$ for $\simplanguage$ as $\after_{\simplanguage}^U(S)=\unfold_{\simplanguage}(\after^U(S))$.
\begin{example}\label{ex:marriage-after}
Consider again a database containing information about books, where
the binary predicate $b$ contains the ISBN (first argument) and the title (second argument). We expect for this database
updates of the form $U = \{b(X,Y)\Leftarrow b^U(X,Y)\}$, where $b^U$
is a query defined by the predicate definition $b^U(X,Y) \leftarrow
b(X,Y)$ $\lor (X =
\param i \land Y = \param t)$, i.e., $U$ is the addition of the tuple $\langle\param i,\param t\rangle$ to $b$. The following integrity
constraint is given:
$$\phi =\; \leftarrow b(X,Y) \land b(X,Z) \land Y \neq Z$$
meaning that no ISBN can be associated with two different titles.
First, each occurrence of $b$
is replaced by $b^U$, obtaining
$$\leftarrow b^U(X,Y) \land b^U(X,Z) \land Y \neq Z.$$
Then $\unfold_{\simplanguage}$ is applied to this integrity
constraint. The first step of definition \ref{def:unfolding}
generates the following:
$$\{\leftarrow (b(X,Y) \lor (X = \param i \land Y = \param t)) \land
(b(X,Z) \lor (X = \param i \land Z = \param t)) \land Y \neq Z\}.$$
The second step translates it to clausal form:
$$\begin{array}{rll}
\after_{\simplanguage}^U(\{\phi\}) = \{
    & \leftarrow b(X,Y) \land b(X,Z) \land  Y \neq Z,\\
    & \leftarrow b(X,Y) \land X = \param i  \land Z = \param t \land Y \neq Z,\\
    & \leftarrow X = \param i \land Y = \param t \land b(X, Z) \land Y \neq
    Z,\\
    & \leftarrow X = \param i \land Y = \param t \land X = \param i \land Z = \param t \land Y \neq Z &
    \}.
\end{array}$$
\end{example}

\subsection{Simplification in
$\simplanguage$}\label{sec:simp-in-simplanguage}

The result returned by $\after_{\simplanguage}$ may contain
redundant parts (e.g., $a= a$) and does not exploit the consistency assumption.
For this purpose, we define a transformation
$\optimize_{\simplanguage}$ that optimizes a given constraint theory
using a set of trusted hypotheses.
We describe here an implementation in terms of sound and
terminating rewrite rules.
Among the tools we use are subsumption, reduction, and resolution.
\begin{definition}[Subsumption]\label{def:subsumption}
Given two denials $D_1$ and $D_2$, $D_1$
\emph{subsumes}\index{subsumption} $D_2$ (via $\sigma$) iff there is a substitution $\sigma$ such that each
literal in $D_1\sigma$ occurs in $D_2$. The subsumption is
\emph{strict}\index{subsumption!strict subsumption} if $D_1$ is not
a variant of $D_2$.
\end{definition}
The subsumption algorithm (see, e.g., \cite{GM92}), besides checking
subsumption, also returns the substitution $\sigma$.
As is well known, the subsuming denial implies the subsumed one.
\begin{example}
	  $\leftarrow p(X,Y)\land q(Y)$
	  subsumes $\leftarrow p(X,b)\land X \neq a \land q(b)$ via $\{Y/b\}$.
\end{example}

\emph{Reduction}~\cite{GM90}
characterizes the elimination of redundancies within a single
denial.
\begin{definition}[Reduction]\label{def:reduction}
For a denial $\phi$, the reduction $\phi^-$ of $\phi$ is the result
of applying on $\phi$ the following rules as long as possible, where
$L$ is a literal,
$c_1,c_2$ are distinct constants, $X$ a bound
variable, $t$ a term, $A$ an atom, $C$, $D$ (possibly empty)
conjunctions of literals and $\vars$\index{1vars@$\vars$} indicates
the set of bound variables occurring in its argument.
\index{$\rewritenospace$ (rewrite)}
$$
\begin{array}{rcll}
\leftarrow L \land C&
\rewrite &
\leftarrow C\mbox{ \;\;if } L \mbox{ is of the form } t=t \mbox{ or } c_1\neq c_2 \\
\leftarrow L \land C&
\rewrite &
\true\mbox{ \;\;if } L \mbox{ is of the form } t\neq t \mbox{ or } c_1= c_2 \\
\leftarrow X = t \land C&
\rewrite &
\leftarrow C\{X/t\}\\
\leftarrow A \land \lnot A \land C&
\rewrite &
\true\\
\leftarrow C \land D&
\rewrite &
\leftarrow D
\mbox{ \;\;if $\leftarrow C$ subsumes $\leftarrow D$ with \!a \!substitution $\sigma$ \!s.t. }\\ && \dom(\sigma) \cap \vars(D) = \emptyset\\
\end{array}
$$
\end{definition}
Clearly, $\phi^- \equiv \phi$.
The last rewrite rule is called subsumption factoring~\cite{EO93}
and includes the elimination of duplicate literals from a denial as
a special case. 
An additional rule for handling parameters may be considered in the
reduction process:
\begin{equation}\label{eqn:reduction-param}
\leftarrow \param a = c_1 \land C \;
\rewrite \;
\leftarrow \param a = c_1 \land C\{\param a/c_1\}
\end{equation}
This may replace parameters with constants and possibly allow
further reduction. For example, the denial $\leftarrow \param a = c
\land b = \param a$ would be transformed into $\leftarrow
\param a = c \land b = c$, and, thus, into $\true$, since $b$ and
$c$ are different constants.

We now briefly recall the definition of resolvent and derivation for
the well-known principle of \emph{resolution}~\cite{Robinson65,ChangLee73} and cast it to the context of deductive databases.
\begin{definition}\label{def:binary-resolvent}
Let $\phi_1'=\;\leftarrow L_1\land\dots\land
L_m,\phi_2'=\;\leftarrow M_1\land\dots\land M_n$ be two standardized
apart variants of denials $\phi_1,\phi_2$. If $\theta$ is a mgu of
$\{L_i,\lnot M_j\}$ then the clause
$$\leftarrow (L_1\land\dots\land L_{i-1}\land L_{i+1}\land\dots L_m \land M_1\land\dots\land M_{j-1}\land M_{j+1}\land\dots \land M_n)\theta$$
is a \emph{binary resolvent} of $\phi_1$ and $\phi_2$ and
$L_i,M_j$ are said to be the literals \emph{resolved upon}.
\end{definition}
The resolution principle is a sound inference rule, in that a
resolvent is a logical consequence of its parent clauses, and preserves range restriction~\cite{DBLP:journals/tcs/Topor87}.

We also refer to the notion of \textit{expansion}\index{expansion}
\cite{CGM90}: the expansion of a clause consists in replacing every
constant or parameter in a database predicate (or variable already
appearing elsewhere in database predicates) by a new variable and
adding the equality between the new variable and the replaced item.
We indicate the expansion of a (set of) denial(s) with a ``+''
superscript.
\begin{example}
Let $\phi=\;\leftarrow p(X,a,X)$. Then $\phi^+=\;\leftarrow p(X,Y,Z)
\land Y = a \land Z = X$.
\end{example}
For a constraint theory $\Gamma$ in
$\simplanguage$ and a denial $\phi$ in $\simplanguage$, 
the notation
$\Gamma \provfbr \phi$ indicates that there is a resolution derivation of a denial
$\psi$ from $\Gamma^+$
and $\psi^-$ subsumes $\phi$.
$\provfbr$ is sound and terminates on any input provided that in each resolution step the
resolvent has at most as many literals as those in the largest denial in $\Gamma^+$.

We can now provide a possible implementation of an optimization operator that eliminates redundant literals and denials from a
given constraint theory $\Gamma$ assuming that another theory
$\Delta$ holds.
Informally, $\provfbr$, subsumption and
reduction are used to approximate entailment.
In the following, $A\sqcup B$ indicates union of
disjoint sets.
\begin{definition}\label{def:optimization}
Given two constraint theories $\Delta$ and $\Gamma$ in
$\simplanguage$,
$\optimize_{\simplanguage}^\Delta(\Gamma)$\index{1optimizesimplanguage@$\optimize_{\simplanguage}$}
is the result of applying the following rewrite rules on $\Gamma$ as
long as possible. In the following, $\phi$, $\psi$ are denials in
$\simplanguage$, $\Gamma'$ is a constraint theory in
$\simplanguage$.
$$
\!\!\!\!\begin{array}{rcl} \{\phi\} \sqcup \Gamma' &
\rewrite &
\Gamma' \mbox{ if } \phi^- = \true \\
\{\phi\} \sqcup \Gamma' &
\rewrite &
\Gamma' \mbox{ if } (\Gamma' \cup \Delta) \provfbr \phi \\
\{\phi\} \sqcup \Gamma' &
\rewrite &
\{\phi^-\} \cup \Gamma' \mbox{ if } \phi \neq \phi^- \neq \true\\
\{ \phi \} \sqcup \Gamma' &
\rewrite &
\{ \psi^-\} \cup \Gamma' \mbox{ if } (\{\phi\} \sqcup \Gamma' \cup
\Delta) \provfbr \psi
\mbox { and } \psi^- \mbox { strictly subsumes } \phi
\end{array}
$$
\end{definition}
The first two rules attempt the elimination of a whole denial,
whereas the last two try to remove literals from a
denial.
We can now fully define simplification for $\simplanguage$:
\begin{definition}\label{def:simp}
Let $S=\langle\idb,\Gamma\rangle\in\simplanguage$ and $U$ be an
update in $\simplanguage$ with respect to $S$. Let
$\unfold_{\simplanguage}(S)=\langle\emptyset,\Gamma'\rangle$. We
define
$\simp_{\simplanguage}^U(S)=\optimize_{\simplanguage}^{\Gamma'}(\after_{\simplanguage}^U(S)).$
\index{1simpsimplanguage@$\simp_{\simplanguage}$}
\end{definition}
Clearly, $\simp_{\simplanguage}^U(S)$ is a CWP of $S$ with respect to $U$.
\begin{example}\label{ex:marriage-simp}
Consider again Example
\ref{ex:marriage-after}.
The
reduction of each denial in $\after_{\simplanguage}^U(\{\phi\})$
generates the following set.
$$\begin{array}{rllll}
\{
    & \leftarrow & b(X,Y) \land b(X,Z) \land  Y \neq Z,\\
    & \leftarrow & b(\param i,Y) \land Y \neq \param t,\\
    & \leftarrow & b(\param i,Z) \land \param t \neq Z &
    \}.
\end{array}$$
Then, the third denial is removed, as it is subsumed by the second
one;
the first constraint is subsumed by $\phi$ and, thus, removed, so
$\simp_{\simplanguage}^U(\{\phi\}) = \; \{ \leftarrow b(\param i,Y)
\land Y \neq \param t\}.$
This result indicates that, for the database to be consistent after
update $U$, book with ISBN $\param i$ must not be already associated with a title
$Y$ different from $\param t$.
\end{example}

\section{Denials with negated existential quantifiers}\label{cha:extensions}

In the definition of $\simplanguage$ we limited the level of
interaction between negation and existential quantification in the
constraint theories. We now relax
this limitation and extend the syntax of denials so as to allow the
presence of negated existential quantifiers. The simplification
procedure can be adapted to such cases, provided that its components
are adjusted as to handle conjuncts starting with a negated
existential quantifier.

In this section we address the problem of simplification of
integrity constraints in \emph{hierarchical} databases, in which all clauses are range
restricted and the schemata are non-recursive, but there is no other restriction on the
occurrence of negation. We refer to the language of such schemata as
$\simpextlanguage$\index{$\simpextlanguage$}.
\begin{definition}[$\simpextlanguage$]\label{def:simp-ext-language}
Let $S$ be a schema. $S$ is in $\simpextlanguage$ if its starred
dependency graph is acyclic.
\end{definition}
Unfolding predicates in integrity constraints with respect to their
definitions cannot be done in the same way as
$\unfold_{\simplanguage}$. For this purpose, we extend the syntax of
denials so as to allow negated existential quantifiers to occur in
literals.
\begin{definition}[Extended denials]\label{def:ext-denial}
\;\;\;\;A \emph{negated \;\;existential
\;\;expression}\index{negated existential expression} or
\emph{NEE}\index{2nee@NEE (negated existential expression)} is an
expression of the form $\lnot \exists \vec X B$, where $B$ is called
the \emph{body} of the NEE, $\vec X$ are some (possibly all) of the
variables occurring in $B$ and $B$ has the form $L_1 \land
\dots\land L_n$, where each $L_i$ is a \emph{general literal}, i.e., either
a literal or a NEE.

\noindent A formula of the form $\forall \vec X (\leftarrow B)$,
where $B$ is the body of a NEE and $\vec X$ are some (possibly all)
of the free variables in $B$, is called an \emph{extended
denial}\index{extended denial}. When there is no ambiguity on the
variables in $\vec X$, extended denials are simply written
$\leftarrow B$.
\end{definition}
\begin{example}\label{ex:extended-denial}
The formula $\leftarrow parent(X) \land \lnot \exists Y
child\_of(X,Y)$ is an extended denial. It reads as follows: there is
inconsistency if there is a parent $X$ that does not have a child.
Note that this is different from the (non-range restricted) denial
$\leftarrow parent(X) \land \lnot child\_of(X,Y)$, which states that
if $X$ is a parent then all individuals must be his/her children.
\end{example}
We observe that variables under a negated existential quantifier
conform with the intuition behind safeness, so we could conclude
that the first formula in example \ref{ex:extended-denial} is safe,
whereas the second one is not.
We can now apply unfolding in $\simpextlanguage$ to obtain extended
denials. In doing so, attention needs to be paid when replacing
negated intensional predicates by their definition, since they may
contain non-distinguished variables and thus existential quantifiers
have to be explicitly indicated. As was the case for
$\unfold_{\simplanguage}$, the replacements may result in
disjunctions and negated conjunctions. Therefore, additional steps
are needed to restore the extended denial form.
\begin{definition}\label{def:unfolding-ext}
Let $S=\langle \idb, \Gamma\rangle$ be a database schema in
$\simpextlanguage$. We define
$\unfold_{\simpextlanguage}(S)$\index{1unfoldsimpextlanguage@$\unfold_{\simpextlanguage}$}
as the set of extended denials obtained by iterating the two
following steps as long as possible:
\begin{enumerate}
\item replace, in $\Gamma$, each occurrence of a literal of the form
$\lnot p(\vec t)$ by $\lnot \exists \vec Y F^p \{\vec X / \vec t\}$
and of a literal of the form $p(\vec t)$ by $F^p \{\vec X / \vec
t\}$, where $F^p$ is $p$'s defining formula,
$\vec X$ its head variables and $\vec Y$ its
non-distinguished variables. If no replacement was made, stop;
\item transform the resulting formula into a set of extended denials
according to the following patterns; $\Phi(Arg)$ is an expression
indicating the body of a NEE in which $Arg$ occurs; $\vec X$ and
$\vec Y$ are disjoint sequences of variables:
\begin{itemize}
\item  $\leftarrow A\land(B\lor C)$ is replaced by
$\leftarrow A\land B$ and $\leftarrow A\land C$;
\item  $\leftarrow A\land\lnot(B\lor C)$ is replaced by
$\leftarrow A\land\lnot B\land \lnot C$;
\item  $\leftarrow A\land\lnot(B\land C)$ is replaced by
$\leftarrow A\land \lnot B$ and $\leftarrow A\land \lnot C$;
\item $\leftarrow A \land \lnot \exists \vec X\Phi(\lnot \exists
\vec Y[B\land(C \lor D)])$ is replaced by\\
$\leftarrow A \land \lnot \exists \vec X\Phi(\lnot \exists \vec
Y[B\land C]\land \lnot \exists \vec Y[B\land D])$.
\end{itemize}
\end{enumerate}
\end{definition}
Without loss of generality, we can assume that, for any NEE $N=\lnot
\exists \vec X B$ occurring in an extended denial $\phi$, the
variables $\vec X$ do not occur outside $N$ in $\phi$. This can
simply be obtained by renaming the variables appropriately and we
refer to such an extended denial as
\emph{standardized}\index{extended denial!standardized extended
denial}.
The \emph{level}\index{negated existential
expression!level}\index{level} of a NEE in an extended denial is the
number of NEEs that contain it, plus $1$. The level of a variable
$X$ in a standardized extended denial is the level of the NEE
starting with $\lnot \exists \vec X$, where $X$ is one of the
variables in $\vec X$, or $0$ if there is no such NEE.
The level of an extended denial\index{extended denial!level} is the
maximum level of its NEEs, or $0$ if there is no NEE.
In example \ref{ex:extended-denial}, the extended denial has level
$1$, $X$ has level $0$ and $Y$ has level $1$.

With a slight abuse of notation, we write in the following
$S\equiv\Psi$ (or $\Psi\equiv S$), where $S=\langle\idb,\ic\rangle$
is a schema and $\Psi$ is a set of extended denials, to indicate
that, for every database $D$ based on $\idb$, $D\models \ic$ iff
$D\models \Psi$.
We can now claim the correctness of $\unfold_{\simpextlanguage}$. We
state the following proposition without a proof, since all steps in
definition \ref{def:unfolding-ext} are trivially
equivalence-preserving.
\begin{proposition}\label{pro:unfolding-ext}
Let $S\in\simpextlanguage$. Then
$\unfold_{\simpextlanguage}(S)\equiv S$.
\end{proposition}
Since the variables under a negated existential quantifier conform
with the intuition behind safeness, the unfolding of a schema in
which all the clauses are range restricted yields a set of extended
denials that still are safe in this sense.
We also note that the language of extended denials is very
expressive.
In \cite{LT84}, it was shown that \emph{any} closed formula of the
form $\forall \vec X (\leftarrow B)$, where $B$ is a first-order
formula and $\vec X$ are its open variables, can be equivalently
expressed by a set of Prolog rules plus the denial $\leftarrow
q(\vec X)$, where $q$ is a fresh predicate symbol.
The construction is such that, if $B$ is function-free, then the
resulting rule set is also function-free (no skolemization is
needed), i.e., it is that of a schema in $\simpextlanguage$. The
unfolding of the obtained schema is therefore an equivalent set of
extended denials.

A simplification procedure can now be constructed for extended
denials in a way similar to what was done in $\simplanguage$.
\begin{definition}\label{def:after-ext}
Let $S$ be a schema in $\simpextlanguage$ and $U$ an update. We
define
$\after^U_{\simpextlanguage}(S)$\index{1aftersimpextlanguage@$\after_{\simpextlanguage}$}
as $\unfold_{\simpextlanguage}(\after^U(S))$.
\end{definition}
The optimization step needs to take into account the nesting of NEEs
in extended denials. Besides the elimination of disjunctions within
NEEs, which is performed by $\unfold_{\simpextlanguage}$, we can
also eliminate, from a NEE, equalities and non-equalities referring
to variables of lower level with respect to the NEE.
\begin{definition}\label{def:equality-elimination}
Let $A,B,C$ be (possibly empty) conjunctions of general literals,
$\vec Y, \vec Z$ disjoint (sequences of) variables, $W$ a variable
of level lower than the level of $\vec Z$, and $\Phi(Arg)$ an
expression indicating a NEE in which $Arg$ occurs.
The following rewrite rules are, respectively, the \emph{equality
elimination}\index{equality elimination rule} and \emph{non-equality
elimination}\index{non-equality elimination rule} rules.
$$\begin{array}{ll}
\leftarrow A \land \Phi(\lnot \exists\vec Y[B \land \lnot
\exists\vec Z(C\land W=c)])
\rewrite\\
\leftarrow A \land \Phi(\lnot \exists\vec Y[B\land W=c \land \lnot \exists\vec Z(C)] \land \lnot \exists\vec Y[B \land W\neq c])\\
\end{array}$$
$$\begin{array}{ll}
\leftarrow A \land \Phi(\lnot \exists\vec Y[B \land \lnot
\exists\vec Z(C\land W\neq c)])
\rewrite\\
\leftarrow A \land \Phi(\lnot \exists\vec Y[B\land W\neq c \land \lnot \exists\vec Z(C)] \land \lnot \exists\vec Y[B \land W=c])\\
\end{array}$$
\end{definition}
The above (non-)equality elimination rewrite rules are equivalence
preserving, as stated below.
\begin{proposition}\label{pro:equality-elimination-correct}
Let $\psi$ be an extended denial and $\psi'$ (resp. $\psi''$) be the
extended denial obtained after an application of the equality (resp.
non-equality) elimination rule. Then $\psi\equiv\psi'$ and
$\psi\equiv\psi''$.
\end{proposition}
\begin{proof}
Using the notation of definition \ref{def:equality-elimination}, we
have:
$$\begin{array}{rl} & \lnot \exists\vec Y[B \land \lnot \exists\vec Z(C\land
W=c)]\\
\equiv &
\lnot \exists\vec Y[B \land (W=c \lor W\neq c) \land \lnot
\exists\vec Z(C\land
W=c)]\\
\equiv &
 \lnot \exists\vec Y[B\land W=c \land \lnot \exists\vec
Z(C\land W=c)] \land\\
& \lnot \exists\vec Y[B \land W\neq c\land \lnot \exists\vec
Z(C\land W=c)]\\
\equiv &
 \lnot \exists\vec Y[B\land W=c \land \lnot \exists\vec
Z(C)] \land \lnot \exists\vec Y[B \land W\neq c]\\
\end{array}$$
In the first step, we added the tautological conjunct $(W=c \lor
W\neq c)$. In the second step we used de Morgan's laws in order to
eliminate the disjunction (as in the definition of
$\unfold_{\simpextlanguage}$). The formula $W=c \land \lnot
\exists\vec Z(C\land W=c)$ can be rewritten as $W=c \land \lnot
(\exists\vec Z C\land W=c)$, and then as $W=c \land (\lnot
\exists\vec Z(C) \lor W\neq c)$, which, with a resolution step,
results in the first NEE in the last extended denial. Similarly,
$W\neq c\land \lnot \exists\vec Z(C\land W=c)$ can be rewritten as
$W\neq c\land (\lnot \exists\vec Z(C) \lor W\neq c)$, which
results (by absorption) into the second NEE in the last extended
denial.

The proof is similar for non-equality elimination.
\end{proof}
In case only levels $0$ and $1$ are involved, the rules look
simpler. For example, equality elimination can be conveniently
formulated as follows.
\begin{equation}\label{eqn:level1-equality-elimination}\begin{array}{rcl} \leftarrow B \land \lnot \exists \vec
X[C\land W = c] &
\rewrite
 \begin{array}{rll}\{&\leftarrow B\{W/c\} \land \lnot
 \exists\vec X C\{W/c\}\\
 &\leftarrow B \land W\neq c&\}.\end{array}
\end{array}\end{equation}
Although (non-)equality elimination does not necessarily shorten the
input formula (in fact, it can also lengthen it), it always reduces
the number of literals in higher level NEEs. Therefore, convergence
to termination can still be guaranteed if this rewrite rule is
applied during optimization. Repeated application of such rules
``pushes'' outwards the involved (non-)equalities until they reach
an NEE whose level is the same as the level of the variable in the
(non-)equality. Then, in case of an equality, the usual equality
elimination step of reduction can be applied.
\begin{example}
The following rewrites show the propagation of a variable of level 0
($X$) from level $2$ to level $0$ via two equality eliminations and
one non-equality elimination.

$$
\!\!\!\!\!\!\!\!\!\!\begin{array}{rcllll}
&  & & \leftarrow &p(X) \land \lnot \exists Y\{q(X,Y) \land \lnot \exists Z[r(X,Y,Z) \land X=a]\} \\\\
 & \equiv & & \leftarrow &p(X) \land
   \lnot \exists Y\{q(X,Y) \land X=a \land \lnot \exists Z[r(X,Y,Z)]\} \\ &&& & \land
   \lnot \exists Y[q(X,Y) \land X\neq a]\\\\
& \equiv &
\{ &\leftarrow &p(a) \land \lnot \exists Y[q(a,Y) \land \lnot \exists Z r(a,Y,Z)] \land \lnot \exists Y[q(a,Y) \land a\neq a],\\
&& &\leftarrow &p(X) \land X\neq a \land \lnot \exists Y[X\neq a \land q(X,Y)] &\}\\
\\
& \equiv &
\{ & \leftarrow &p(a) \land \lnot \exists Y[q(a,Y) \land \lnot \exists Z r(a,Y,Z)]\land \lnot \exists Y[q(a,Y) \land a\neq a],\\
&&& \leftarrow &p(X) \land X\neq a \land \lnot \exists Y q(X,Y),\\
&&& \leftarrow &p(a) \land X\neq a \land X= a&\}\\
\\
& \equiv &
\{ & \leftarrow &p(a) \land \lnot \exists Y[q(a,Y) \land \lnot \exists Z r(a,Y,Z)],\\
&&& \leftarrow &p(X) \land X\neq a \land \lnot \exists Y q(X,Y)&\}\\
\end{array}
$$
In the first step we applied equality elimination to $X=a$ at level
$2$.
In the second step we applied the rewrite rule
\eqref{eqn:level1-equality-elimination} for equality elimination to
$X=a$ at level 1.
Then we applied non-equality elimination to $X\neq a$ at level $1$
in the second denial.
In the last step, we removed, by standard application of reduction,
the last extended denial and the last NEE in the first extended
denial, which are clearly tautological.
\end{example}
We observe that the body of a NEE is structurally similar to the
body of an extended denial. The only difference is that, in the
former, there are variables that are quantified at a lower level.
According to this observation, such (free) variables in the body of
a NEE are to be treated as parameters during the different
optimization steps, since, as was indicated on page
\pageref{page:parameter-introduced}, parameters are free variables.

Reduction (Definition~\ref{def:reduction}) can then take place in NEE bodies exactly
as in ordinary denials, with the proviso above of treating free
variables as parameters.

The definition of resolution (Definition~\ref{def:binary-resolvent}) can be
adapted for extended denials by applying it to general literals
instead of literals.

Subsumption (Definition \ref{def:subsumption}) can also be applied to extended denials
without changing the definition. However, we can slightly modify the
notion of subsumption to explore the different levels of NEEs in an
extended denial. This is captured by the following definition.
\begin{definition}\label{def:extended-subsumption}\index{subsumption!extended subsumption}
Let $\phi=\leftarrow A\land B$ and $\psi=\leftarrow C\land D$ be two
extended denials, where $A$ and $C$ are (possibly empty)
conjunctions of literals and $B$ and $D$ are (possibly empty)
conjunctions of NEEs. Then $\phi$ extended-subsumes $\psi$ if both
conditions (1) and (2) below hold.
\begin{enumerate}
\item[(1)] $\leftarrow A$ subsumes $\leftarrow C$ with substitution
$\sigma$.
\end{enumerate}
\begin{enumerate}
\item[(2)] For every NEE $\lnot\exists \vec X_N N$ in $B$, there is
a NEE $\lnot\exists \vec X_M M$ in $D$ such that $\leftarrow
N\sigma$
 is extended-subsumed by $\leftarrow M$.
\end{enumerate}
\end{definition}
\begin{example}
The extended denial $\leftarrow p(X)\land \lnot \exists Y,Z
[q(X,Y)\land r(Y,Z)]$ extended-subsumes the extended denial
$\leftarrow p(a)\land \lnot \exists T [q(a,T)] \land \lnot \exists W
[s(T)]$, since $\leftarrow p(X)$ subsumes $\leftarrow p(a)$ with
substitution $\{X/a\}$ and, in turn, $\leftarrow q(a,T)$ subsumes
$\leftarrow q(a,Y)\land r(Y,Z)$.
\end{example}
This definition encompasses ordinary subsumption, in that it
coincides with it if $B$ and $D$ are empty. Furthermore, it captures
the desired property that if $\phi$ extended-subsumes $\psi$, then
$\phi\models \psi$; the reverse, as in subsumption, does not
necessarily hold.
\begin{proposition}
Let $\phi$ and $\psi$ be extended denials. If $\phi$
extended-subsumes $\psi$ then $\phi\models\psi$.
\end{proposition}
\begin{proof}
Let $\phi,\psi,A,B,C,D$ be as in Definition
\ref{def:extended-subsumption}. If $B$ and $D$ are empty the claim
holds, since $\phi$ and $\psi$ are ordinary denials.
The claim also holds if $D$ is not empty, since $\leftarrow C$
entails $\leftarrow C\land D$.
We now show the general claim with an inductive proof on the level
of extended denials.

\noindent The base case (level $0$) is already proven.

\noindent Inductive step. Suppose now that $\phi$ is of level $n+1$
and that the claim holds for extended denials of level $n$ or less.
Assume as a first case that $B$ is empty. Then $\phi$ entails
$\psi$, since $\phi$ entails $\leftarrow C$ ($\leftarrow A$ subsumes
$\leftarrow C$ by hypothesis).
Assume for the moment that $B=\lnot\exists \vec X_N N$ is a NEE of
level $1$ in $\phi$ and that $D$ contains a NEE (of level $1$)
$\lnot\exists \vec X_M M$, such that $\phi'=\leftarrow N\sigma$%
is extended-subsumed by $\psi'=\leftarrow M$, as assumed in the hypotheses.
But $\phi'$ and $\psi'$
are extended denials of level $n$ and, therefore, if $\psi'$
subsumes $\phi'$ then $\psi'$ entails $\phi'$ by inductive
hypothesis.
Clearly, since $\leftarrow A$ entails $\leftarrow C\land D$ (by
hypothesis) and $\leftarrow M$%
 entails $\leftarrow N\sigma$%
(as a consequence of the inductive hypothesis), then $\leftarrow
A\land B$ entails $\leftarrow C\land D$, which is our claim. If $B$
contains more than one NEE of level $1$, the argument is iterated by
adding one NEE at a time.
\end{proof}
The inductive proof also shows how to check extended subsumption
with a finite number of subsumption tests. This implies that
extended subsumption is decidable, since subsumption is.
Now that a correct extended subsumption is introduced, it can be
used instead of subsumption in the subsumption factoring rule of
reduction (Definition~\ref{def:reduction}).
In the following, when referring to a NEE $N=\lnot \exists \vec X B$
we can also write it as a denial $\leftarrow B$, with the
understanding that the free variables in $N$ are considered
parameters.
\begin{definition}\label{def:reduction-ext}\index{reduction!for extended denials}
For an extended denial $\phi$, the reduction $\phi^-$ of $\phi$ is
the result of applying on $\phi$ equality and non-equality
elimination as long as possible, and then the rules of Definition~\ref{def:reduction} (reduction) on
$\phi$ and its NEEs as long as possible, where ``literal'' is
replaced by ``general literal'', ``subsumes'' by
``extended-subsumes'' and ``denial'' by ``extended denial''.
\end{definition}
Without reintroducing similar definitions, we assume that the same
word replacements are made for the notion of $\provfbr$. The underlying
notions of substitution and unification also apply to extended
denials and general literals; only, after substitution with a
constant or parameter, the existential quantifier of a variable is
removed. For example, the extended denials $\phi=\leftarrow
p(X,b)\land \lnot\exists Z[q(Z,X)]$ and $\psi=\leftarrow p(a,Y)\land
\lnot q(c,a)$ unify with substitution $\{X/a, Y/b, Z/c\}$.
By virtue of the similarity between denial bodies and NEE bodies, we
extend the notion of optimization as follows.
\begin{definition}\label{def:optimization-ext}
Given two sets of extended denials $\Delta$ and $\Gamma$,
$\optimize_{\simpextlanguage}^\Delta(\Gamma)$\index{1optimizesimpextlanguage@$\optimize_{\simpextlanguage}$}
is the result of applying the following rewrite rules and the rules
of Definition \ref{def:optimization} ($\optimize_{\simplanguage}$) on $\Gamma$ as long as possible. In
the following, $\phi$ and $\psi$ are NEEs, $\Gamma'$ is a set of
extended denials, and $\Phi(Arg)$ is an expression indicating the
body of an extended denial in which $Arg$ occurs.
$$
\!\!\!\!\begin{array}{rcl} \{\leftarrow \Phi(\phi)\} \sqcup \Gamma' &
\rewrite &
\{\leftarrow \Phi(\true) \} \cup \Gamma' \mbox{ if } \phi^- = \true \\
\{\leftarrow \Phi(\phi)\} \sqcup \Gamma' &
\rewrite &
\{\leftarrow \Phi(\true) \} \cup \Gamma' \mbox{ if } (\Gamma' \cup \Delta) \provfbr \phi \\
\{\leftarrow \Phi(\phi)\} \sqcup \Gamma' &
\rewrite &
\{\leftarrow \Phi(\phi^-)\} \cup \Gamma' \mbox{ if } \phi \neq \phi^- \neq \true\\
\{\leftarrow \Phi(\phi) \} \sqcup \Gamma' &
\rewrite &
\{\leftarrow \Phi(\phi^-)\} \cup \Gamma' \mbox{ if } ((\{\phi\} \cup
\{\leftarrow\Phi(\phi)\})\sqcup \Gamma' \cup \Delta) \provfbr \psi\\
&&\mbox { and } \psi^- \mbox { strictly extended-subsumes } \phi
\end{array}
$$
\end{definition}
Finally, the simplification procedure for $\simpextlanguage$ is
composed in terms of $\after_{\simpextlanguage}$ and
$\optimize_{\simpextlanguage}$.
\begin{definition}\label{def:simp-ext}
Consider a schema $S=\langle\idb,\Gamma\rangle\in\simpextlanguage$
and an update $U$. Let $\Gamma'=$ $\unfold_{\simpextlanguage}(S)$.
We define
$$\simp_{\simpextlanguage}^U(S)=\optimize_{\simpextlanguage}^{\Gamma'}(\after_{\simpextlanguage}^U(S)).$$
\index{1simpsimpextlanguage@$\simp_{\simpextlanguage}$}
\end{definition}
Similarly to $\simplanguage$, soundness of the optimization steps
and the fact that $\after$ returns a WP entail the following.
\begin{proposition}
Let $S\in\simpextlanguage$ and $U$ be an update. Then
$\simp_{\simpextlanguage}^U(S)$ is a CWP of $S$ with respect to $U$.
\end{proposition}

\subsection{Examples}

We now discuss the most complex non-recursive examples that we found
in the literature for testing the effectiveness of the proposed simplification procedure.
\begin{example}
This example is taken from \cite{LL96}.
Consider a schema $S=\langle \idb,\Gamma\rangle$ with three
extensional predicates $a$, $b$, $c$, two intensional predicates
$p$, $q$, a constraint theory $\Gamma$ and a set of trusted
hypotheses $\Delta$.
$$\begin{array}{rll}
\idb = \{&p(X,Y) \leftarrow a(X,Z) \land b(Z,Y),\\
&q(X,Y)  \leftarrow p(X,Z) \land c(Z,Y)&\}\\
\Gamma = \{& \leftarrow p(X,X) \land \lnot q(1,X)\;\;\}\\
\Delta = \{& \leftarrow a(1,5)\;\;\}
\end{array}$$
This schema $S$ is not in $\simplanguage$ and the unfolding of $S$
is as follows.
$$\unfold_{\simpextlanguage}(S) = \{\leftarrow a(X,Y) \land b(Y,X) \land \lnot \exists W,Z (a(1,W)
\land b(W,Z) \land c(Z,X))\}.$$
We want to verify that the update $U = \{ b(X,Y) \Leftarrow b(X,Y)
\land X\neq 5\}$ (the deletion of all $b$-tuples in which the first
argument is $5$) does not affect consistency.
$\after_{\simpextlanguage}^U(S)$ results in the following extended
denial:
$$\leftarrow a(X,Y) \land b(Y,X) \land Y \neq 5 \land \lnot \exists W,Z (a(1,W)
\land b(W,Z) \land W \neq 5 \land c(Z,X)).$$
As previously described, during the optimization process, the last
conjunct can be processed as a separate denial $\phi = \;\leftarrow
a(1,W) \land b(W,Z) \land W \neq 5 \land c(Z,\param X)$, where
$\param X$ is a free variable that can be treated as a parameter
(and thus indicated in bold). With a resolution step with $\Delta$,
the literal $W \neq 5$ is proved to be redundant and can thus be
removed from $\phi$. The obtained formula is then subsumed by
$\unfold_{\simpextlanguage}(S)$ and therefore
$\simp_{\simpextlanguage}^U(S)=\emptyset$, i.e.,
the update cannot violate the integrity constraint, which is the
same result that was found in \cite{LL96}.
\end{example}

In order to simplify the notation for tuple additions and deletions,
we write $p(\vec\param a)$ as a shorthand for the database update
$p(\vec X) \Leftarrow p(\vec X)\lor \vec X= \vec\param a$ and $\lnot
p(\vec\param a)$ for $p(\vec X) \Leftarrow p(\vec X)\land \vec X\neq
\vec\param a$\label{notation:simplified-update-notation}.

\begin{example}\label{ex:leuschel}
The following schema $S$ is the relevant part of an example
described in \cite{LD98} on page 24.

$$\begin{array}{rrcll}
S=\langle\{&married\_to(X, Y) & \leftarrow &
   parent(X, Z) \land parent(Y, Z) \land\\&&&
   man(X) \land woman(Y),\\
&married\_man(X) & \leftarrow & married\_to(X, Y),\\
&married\_woman(X) & \leftarrow &
married\_to(Y, X),\\
&unmarried(X) & \leftarrow & man(X) \land \lnot married\_man(X),\\
&unmarried(X) & \leftarrow & woman(X) \land \lnot
married\_woman(X)&\},\\
 \{&& \leftarrow & man(X) \land
woman(X),\\
&& \leftarrow & parent(X, Y) \land unmarried(X)&\}\rangle
\end{array}$$
If we reformulate the example using the shorthand notation,
the database is updated with $U = \{man(\param a)\}$, where $\param
a$ is a parameter. The unfolding given by
$\unfold_{\simpextlanguage}(S)$ is as follows, where $m$, $w$, $p$
respectively abbreviate $man$, $woman$, $parent$, which are the only
extensional predicates.
$$\begin{array}{rll}\{

& \leftarrow
m(X)\land  w(X),\\
& \leftarrow p(X, Y)\land m(X)\land \lnot \exists(T,Z) [p(X, Z)\land
p(T, Z)\land
m(X)\land w(T)],\\
& \leftarrow p(X, Y)\land w(X)\land \lnot \exists(T,Z) [p(X, Z)\land
p(T, Z)\land w(X)\land m(T)]&\}
\end{array}
$$
We start the simplification process by applying
$\after_{\simpextlanguage}$ to $S$ wrt $U$.
$$\begin{array}{rlll}\after_{\simpextlanguage}^U(S) \equiv \{&
\leftarrow &(m(X) \lor X = \param a) \land w(X),\\
&\leftarrow &p(X, Y) \land (m(X) \lor X=\param a)\land \\&&\lnot
\exists(T,Z) [p(X,
Z)\land p(T, Z)\land (m(X) \lor X=\param a)\land w(T)],\\
&\leftarrow &p(X, Y)\land w(X)\land \\&&\lnot \exists(T,Z) [p(X,
Z)\land p(T, Z)\land w(X)\land (m(T)\lor T=\param a)]
&\!\!\!\!\!\!\}
\end{array}$$
After eliminating the disjunctions at level $0$,
$\after_{\simpextlanguage}^U(S)$ is as follows:
$$\!\!\!\begin{array}{rlll}\{
\leftarrow &\!\!\!\!\!m(X) \land w(X),\\
\leftarrow &\!\!\!\!\!X = \param a \land w(X),\\
\leftarrow &\!\!\!\!\!p(X, Y) \land m(X)
 \land \lnot \exists(T,Z) [p(X, Z)\land p(T, Z)\land (m(X) \lor
 X=\param a)\land w(T)],\\
\leftarrow &\!\!\!\!\!p(X, Y) \land X=\param a\land \lnot
\exists(T,Z) [p(X, Z)\land p(T, Z)\land (m(X) \lor
X=\param a)\land w(T)],\\
\leftarrow &\!\!\!\!\!p(X, Y)\land w(X)\land \lnot \exists(T,Z)
[p(X, Z)\land p(T, Z)\land w(X)\land (m(T)\lor T=\param a)]
&\!\!\!\!\!\!\}
\end{array}$$
Now we can eliminate the disjunctions at level $1$ and obtain the
following set.
$$\begin{array}{rlll}\{
\leftarrow &m(X) \land w(X),\\
\leftarrow &X = \param a \land w(X),\\
\leftarrow &p(X, Y) \land m(X)
 \land \lnot \exists(T,Z) [p(X, Z)\land p(T, Z)\land m(X)
\land w(T)] \land \\ &\lnot \exists(T,Z) [p(X, Z)\land p(T, Z)\land
 X=\param a \land w(T)],\\
\leftarrow &p(X, Y) \land X=\param a\land\lnot \exists(T,Z) [p(X,
Z)\land
 p(T, Z)\land m(X)
\land w(T)] \land \\ &\lnot \exists(T,Z) [p(X, Z)\land
 p(T, Z)\land
 X=\param a \land w(T)],\\
\leftarrow &p(X, Y)\land w(X)\land \lnot \exists(T,Z) [p(X, Z)\land
p(T,Z)\land w(X)\land m(T)] \land\\
& \lnot \exists(T,Z) [p(X, Z)\land p(T,Z)\land w(X)\land T=\param
a]&\}
\end{array}$$
We can now proceed with the optimization of this set of extended
denials by using the $\optimize_{\simpextlanguage}$ transformation.
Clearly, the first, the third and the fifth extended denial are
extended-subsumed by the first, the second and, respectively, the
third extended denial in $\unfold_{\simpextlanguage}(S)$ and are
thus eliminated. The second denial reduces to $\leftarrow w(\param
a)$. In the fourth denial the equality $X=\param a$ at level $0$ is
eliminated, thus substituting $X$ with $\param a$ in the whole
extended denial.
We obtain the following.
$$\begin{array}{rlll}
\{ \leftarrow &w(\param a),\\
 \leftarrow &p(\param a, Y) \land \lnot \exists(T,Z) [p(\param a, Z)\land
 p(T, Z)\land m(\param a)
\land w(T)] \land \\
&\lnot \exists(T,Z) [p(\param a, Z)\land
 p(T, Z)\land \param a = \param a \land w(T)]&\}\\
\end{array}$$
For the last extended denial, first we can eliminate the trivially
succeeding equality $\param a=\param a$ from the body of the second
NEE. Then we can consider that
$$\leftarrow \lnot \exists(T,Z) [p(\param a, Z)\land
 p(T, Z)\land m(\param a)
\land w(T)]$$
extended-subsumes
$$%
\leftarrow p(\param a, Y) \land \lnot \exists(T,Z) [p(\param a, Z)\land
 p(T, Z)\land w(T)]$$
so we can eliminate by subsumption factoring the subsuming part and
leave the subsumed one.
The simplification procedure for $\simpextlanguage$ applied to $S$
and $U$ returns the following result.
$$\begin{array}{rll}\simp_{\simpextlanguage}^U(S) = \{
& \leftarrow
w(\param a),\\
& \leftarrow p(\param a, Y)\land \lnot \exists(T,Z) [p(\param a,
Z)\land p(T, Z)\land w(T)]&\}
\end{array}$$
This coincides with the result given in \cite{LD98}, rewritten with
our notation, with the only difference that they do not assume
disjointness of $\idb$ and $\edb$, so, in the latter extended
denial, they have the extra conjunct $\lnot \exists
V[married\_to(\param a,V)]$.
\end{example}

\section{Related work}\label{sec:related}

The idea of simplifying integrity constraints has been long recognized, dating back to at least
\cite{Nic82,DBLP:conf/sigmod/BernsteinB82}, and then elaborated by
several other authors, e.g.,
\cite{HMN84,DBLP:conf/sigmod/HsuI85,LST87,Qia88,SK88,CGM90,DC94,LL96,LD98,
SS99,Dec02}.
We continued this line of research in~\cite{M:PHD2005,DBLP:conf/lopstr/ChristiansenM03}.
In particular, in~\cite{M:PHD2005} we considered an expressive language, $\simpextlanguage$, for formulating schemata and constraints, for which we could provide a sound and terminating simplification procedure.

The simplification problem is also tightly connected to Query Containment (QC)~\cite{chandra1977optimal}, i.e., the problem of establishing whether, given two queries $q_1$ and $q_2$, the answer to $q_1$ is always a subset of the answer to $q_2$, for all database instances. Indeed, ``perfect'' simplifications may be generated if and only if QC is decidable. Unfortunately, QC is already undecidable for $\datalog$ databases without negation~\cite{shmueli87,DBLP:conf/er/CaliM08}.

The kind of integrity constraints considered here
are the so-called \emph{static} integrity constraints, in that they refer to 
properties that must be met by the data in each database state.
\emph{Dynamic} constraints, instead,
are used to impose restrictions on the way the database states can
evolve over time, be it on successive states or on arbitrary sets of states.
Dynamic constraints have been considered, e.g., in
\cite{DBLP:journals/tods/Chomicki95,DBLP:conf/vldb/CowleyP00}.

Another important distinction can be made between
\emph{hard} (or
\emph{strong})
constraints and \emph{soft}
(or \emph{deontic})
 constraints.
The former ones are used to model necessary requirements of the
world that the database represents, like the constraints studied here.
Deontic constraints govern what is obligatory but not necessary of
the world, so violations of deontic constraints correspond to
violations of obligations that have to be reported to the user or
administrator of the database rather than inconsistencies proper.
The wording ``soft constraint'' sometimes refers to a condition that
should preferably be satisfied, but may also be violated. Deontic constraints have
been considered, e.g., in \cite{DBLP:journals/fuin/CarmoDJ01} and
soft constraints in \cite{375749}.
Soft constraints are akin to preferences, covered by an abundant literature~\cite{DBLP:journals/jacm/CiacciaMT20}.

Other kinds of constraints, of a more structural nature, are the so-called access constraints, or access patterns, whose interaction with integrity constraints proper has also been studied~\cite{DBLP:conf/icdt/BaranyBB13}.

Often, integrity constraints are characterized as \emph{dependencies}. Many common dependencies are used in database practice. Among those involving a single relation, we mention \emph{functional dependencies}~\cite{Cod72} (including \emph{key dependencies}) and \emph{multi-valued dependencies}~\cite{DBLP:journals/tods/Fagin77,DBLP:journals/jacm/Fagin82}. For the inter-relation case, the most common ones are \emph{inclusion dependencies} (i.e., \emph{referential
constraints}).
All of these can straightforwardly be represented as (extended denials).
In particular, many authors acknowledge tuple-generating dependencies (TGDs) and equality-generating dependencies (EGDs) to be the most important types of dependencies in a database, since they encompass most others~\cite{DBLP:journals/jacm/BeeriV84} and are also commonly used for schema mappings in Data Exchange~\cite{DBLP:journals/tcs/FaginKMP05}.
A TGD is a formula of the form $\forall \vec X(\phi(\vec X)\rightarrow\exists \vec Y\psi(\vec X,\vec Y))$, where $\phi(\vec X)$ is a conjunction of atomic formulas, all with variables among the variables in $\vec X$; every variable in $\vec X$ appears in $\phi(\vec X)$ (but not necessarily in $\psi$); $\psi(\vec X, \vec Y)$ is a conjunction of atoms, all with variables among $\vec X$ and $\vec Y$.
Clearly, such a TGD can be expressed with our notation as the extended denial $\leftarrow \phi(\vec X)\land\lnot\exists\vec Y \psi(\vec X,\vec Y)$.
EGDs are formulas of the form $\forall \vec X(\phi(\vec X)\rightarrow X_1=X_2)$, where $X_1$ and $X_2$ are variables in $\vec X$. Clearly, this is expressed as a denial $\leftarrow \phi(\vec X)\land X_1\neq X_2$.
So our framework captures the most representative cases of constraints used in the literature.

More detailed classifications of integrity constraints for deductive
databases have been attempted by several authors, e.g.,
\cite{DBLP:books/cs/Ullman88,DBLP:conf/vldb/Grefen93}.

Once illegal updates are detected, it must be decided how to restore
a consistent database state. We targeted a ``prevention'' approach, that avoids illegal updates completely. Other approaches, instead, need to restore consistency via corrective actions after an illegal update -- typically a rollback, but several works compute a r\emph{repair} that changes, adds or deletes tuples of the database in order to satisfy the integrity constraints again. 
The generation of repairs is a nontrivial
issue; see, e.g.,
\cite{DBLP:conf/pods/ArenasBC99,773179,DBLP:conf/foiks/ArieliDNB04}, and~\cite{DBLP:series/synthesis/2011Bertossi,DBLP:conf/pods/Bertossi19} for surveys on the topic.

In some scenarios, a temporary violation of integrity constraints
may be accepted provided that consistency is quickly repaired; if,
by nature of the database application, the data are particularly
unreliable, inconsistencies may be even considered unavoidable.
Approaches that cope with the presence of inconsistencies have been studied and gave rise to \emph{inconsistency-tolerant integrity checking}~\cite{DBLP:conf/lpar/DeckerM06,DBLP:conf/dexaw/DeckerM06,DBLP:conf/dexaw/DeckerM07,DBLP:conf/ppdp/DeckerM08}.
Besides being checked and tolerated, inconsistency can also be measured, and numerous indicators have been studied to address this problem~\cite{DBLP:journals/jiis/GrantH06,DBLP:conf/er/DeckerM09,DBLP:conf/ijcai/GrantH11,DBLP:conf/ecsqaru/GrantH13,DBLP:journals/ijar/GrantH17,DBLP:journals/ijar/GrantH23}.

An orthogonal research avenue is that of allowing inconsistencies to occur in databases but to
filter the data during query processing so as to provide a
\emph{consistent query answer}~\cite{DBLP:conf/pods/ArenasBC99}, i.e., the set of tuples that
answer a query in all possible repairs (of course without actually
having to compute all the repairs).

Integrity checking has been applied to a number of different contexts, including data integration~\cite{DBLP:conf/foiks/ChristiansenM04}, the presence of aggregates~\cite{DBLP:conf/fqas/Martinenghi04,DBLP:journals/access/SamarinA21}, 
the interaction with transaction management~\cite{DBLP:conf/dexa/MartinenghiC05},
the use of symbolic constraints~\cite{DBLP:journals/aai/ChristiansenM00},
approximate constraints~\cite{DBLP:conf/icdt/KenigS20}
big data~\cite{DBLP:journals/tnse/YuHYL21} and data clouds~\cite{DBLP:conf/wasa/YangTWZ20}.

We also observe that integrity checking is also commonly included as a typical part of complex data preparation pipelines for subsequent processing based, e.g., on Machine Learning or Clustering algorithms~\cite{DBLP:conf/fqas/Masciari09,DBLP:journals/isci/MasciariMZ14,DBLP:conf/ideas/FazzingaFMF09,DBLP:journals/tods/FazzingaFFM13} as well as crowdsourcing applications~\cite{DBLP:conf/socialcom/GalliFMTN12,DBLP:conf/www/BozzonCCFMT12}.

\section{Conclusion}

We applied program transformation operators to the generation of simplified integrity constraints, targeting an expressive class termed extended denials, which includes negated existential quantifiers.
We believe that this is an important class, as it encompasses very common dependencies such as tuple-generating dependencies and equality-generating dependencies.

An immediate application of our operators is an automated process that, upon requests from an application, communicates with the database and transparently carries out the required simplified
integrity checking operations. The would imply benefits in
terms of efficiency and could leverage a compiled approach, since
simplifications can be generated at design time.

Although we used a logical notation throughout the thesis, standard
ways of translating integrity constraints into SQL exist, although
further investigation is needed in order to handle additional
language concepts of SQL like null values.
In \cite{Dec02}, Decker showed how to implement integrity constraint
checking by translating first-order logic specifications into SQL
triggers. The result of our transformations can be combined with
similar translation techniques and thus integrated in an active
database system, according to the idea of embedding integrity
control in triggers.
In this way the advantages of declarativity are combined with the
efficiency of execution.

Other possible enhancements of the proposed framework may be
developed using statistical measures on the data, such as estimated
relation sizes and cost measurements for performing join and union
operations. Work
in this area is closely related to methods for dynamic query
processing, e.g.,
\cite{DBLP:conf/sigmod/ColeG94,DBLP:conf/sigmod/SeshadriHPLRSSS96}.

The proposed procedure is guaranteed to terminate, as we approximated entailment with rewrite rules based on resolution, subsumption and replacement of specific patterns. While~\cite{M:PHD2005} discusses a few cases in which, besides termination, the procedure guarantees also completeness, it would be interesting to pinpoint more specifically what is left out in more expressive cases.

Another line of research regards the combination of extended denials with other expressive scenarios, also individually described in~\cite{M:PHD2005}, such as the addition of aggregates and arithmetic built-ins.
While all these extensions could be trivially handled by a rule set comprising all rewrite rules defined for each specific scenario (such rules are not mutually exclusive), it should be interesting to study whether further improvements can be obtained by exploiting the interaction between these rules.

\bibliographystyle{abbrv}

\end{document}